\documentclass{llncs}
\usepackage{graphicx}
\usepackage{mathptmx}
\usepackage{amsmath}
\usepackage{hyperref}
\usepackage{cite}

\DeclareSymbolFont{largesymbols}{OMX}{cmex}{m}{n}

\makeatletter
\def\hb@xt@{\hbox to }
\makeatother

\let\oldendproof\endproof
\def\endproof{\qed\oldendproof}

\begin{document}
\title{Solving Single-digit Sudoku Subproblems} 

\author{David Eppstein}

\institute{Computer Science Department, University of California, Irvine}

\maketitle   

\pagestyle{plain}

\begin{abstract}
We show that single-digit ``Nishio'' subproblems in $n\times n$ Sudoku puzzles may be solved in time $o(2^n)$, faster than previous solutions such as the pattern overlay method. We also show that single-digit deduction in Sudoku is NP-hard.
\end{abstract}

\section{Introduction}

Sudoku puzzles, appearing daily in newspapers and collected in many books, take the form of a $9\times 9$ grid of cells, some of which are filled in with numbers in the range from 1 to 9, such as the one in  Figure~\ref{fig:example} below.
The task for the puzzle solver is to fill in the remaining cells so that the numbers within every row, column, or $3\times 3$ block of cells (bounded by the heavier border lines) are all distinct.

\begin{figure}[h]
\centering\includegraphics[height=1.75in]{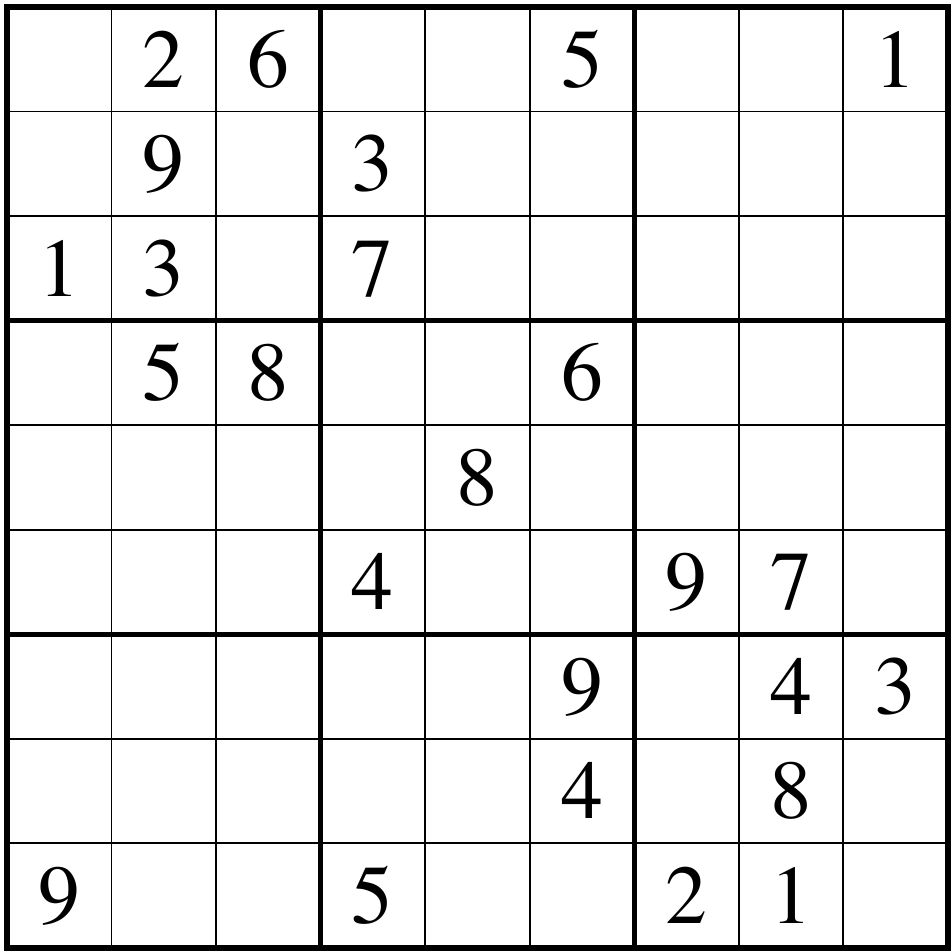}
\caption{A difficult Sudoku puzzle.}
\label{fig:example}
\end{figure}

Every properly-designed Sudoku puzzle has a unique solution, and in most cases this solution can be found by humans via a sequence of deductions that does not involve trial and error or backtracking. Sudoku puzzlists typically learn a repertoire of deductive rules that allow them to make steps towards solving a puzzle whenever they can find certain patterns in it. As an example, if one of the $3\times 3$ blocks can only contain digit $d$ in cells that belong to a single column, then $d$ cannot be placed in any cell that belongs to in that column but not to that block, for doing so would make it impossible to place $d$ within the block.  Thus, in Figure~\ref{fig:example}, the digit 6 must be placed in one of two cells in the middle column of the top middle $3\times 3$ block, and therefore it cannot be placed in the middle column of the bottom middle block. Rules such as this, and exercises to help learn them, may be found in guidebooks for Sudoku puzzle solving; see, e.g.,~\cite{GorLon-06,Ste-07}.

Common variants of Sudoku shrink the puzzle to a $6\times 6$ grid with $2\times 3$ blocks, or expand it to a $16\times 16$ grid with $4\times 4$ blocks; more generally, Sudoku puzzles may use any $ab\times ab$ grid with blocks of size $a\times b$. This variability allows us to apply the tools of computational complexity to Sudoku, and it has been shown that, when generalized to $n\times n$ grids, solving Sudoku is NP-hard~\cite{YatSet-ITF-03}. In practice, however, it is not difficult for computers to solve $9\times 9$ Sudoku puzzles using backtracking together with simple deductive rules to determine the consequences of each choice and prune the search tree whenever an inconsistency is discovered~\cite{Bro-NAW-06,cs.DS/0507053}.  Worst case bounds on the time to solve Sudoku problems may be found by reducing Sudoku to graph list coloring on a graph with $n^2$ vertices, one for each cell, with a color for each number that may be placed within a cell and with edges connecting pairs of cells that belong to the same row, column, or block~\cite{HerMur-NAMS-07,RosTaa-11-Ch7.3}. A coloring algorithm of Lawler~\cite{Law-IPL-76},  specialized to these graphs (using the fact that any color class in a valid coloring may be represented as a permutation) can solve any Sudoku puzzle in worst case time $O(2^{n^2}\cdot n! \cdot n^{O(1)})$. Sudoku puzzles may be reduced to instances of \emph{exact satisfiability} (satisfiability of CNF formulae by an assignment that makes exactly one term in each clause true) with $O(n^2)$ variables and clauses~\cite{OstPar-ICTAI-08} or to instances of exact set cover with $O(n^2)$ sets and elements~\cite{HunPonTuc-UMAP-08}, again achieving a time bound exponential in $n^2$.
Sudoku puzzles have also been used as a test case for many varied computer problem-solving techniques including constraint programming~\cite{Sim-MRCSP-05}, reduction to satisfiability~\cite{LynOua-AIMATH-06,Web-LPAR-05}, harmony search~\cite{Gee-KBIIES-07}, metaheuristics~\cite{Lew-Heu-07}, genetic algorithms~\cite{ManKol-CEC-07}, particle swarms~\cite{MorTog-GECCO-07}, belief propagation~\cite{MooGun-MWALS-06}, projection onto the Birkhoff polytope~\cite{MooGunKup-TIT-09}, Gr\"obner bases~\cite{ArnLucTaa-CMJ-10}, mixed integer programming~\cite{Koc-ORP-05}, and compressed sensing~\cite{BabPelSto-SPL-10}.

Despite the ease of solving Sudoku using these methods, there is still a need for computer puzzle solvers that eschew both trial and error and sophisticated computer search techniques, and that instead mimic human deduction~\cite{Bro-NAW-06,cs.DS/0507053}. With a deductive solver, we may estimate the difficulty of a puzzle by solving it using rules that are as simple as possible and then assessing the cognitive complexity of the rules that were needed. Difficulty estimation techniques of this type may also be used as a subroutine to guide the puzzle generation process~\cite{HunPonTuc-UMAP-08}, or to explain a puzzle's solution to a human, in order to help teach the human how to solve Sudoku puzzles more easily. Additionally, deductive puzzle solving algorithms may cast light on human psychology, by helping us understand what kinds of thought processes are easier or more difficult for humans.

One of the more complex deduction rules commonly used by human Sudoku puzzlers considers the possible positions still open to a single digit $d$, and eliminates possible placements whenever they cannot be part of a consistent placement of $d$ in every row, column, and block of the puzzle, regardless of how the other digits are placed. This technique is often called \emph{Nishio}, after Tetsuya Nishio, although the same word may also apply to unrestricted trial-and-error solution. Nishio deductions on a $9\times 9$ board are within the ability of skilled humans to solve mentally, often without written notes, but it is not obvious how to implement them in a computer simulation of human reasoning. In our prior work on a Sudoku solver~\cite{cs.DS/0507053} we used a polynomial time approximation to Nishio based on graph matching algorithms, that always makes correct deductions but that, by ignoring the block structure of a Sudoku puzzle, is not always able to solve puzzles that can be solved by Nishio such as the one in Figure~\ref{fig:example}.

The main result of this paper is an algorithm for solving this single-digit deduction problem exactly, with running time $o(2^n)$. Although its running time is exponential, our new algorithm is implementable, faster than the $n!/2^{O(n)}$ time of previously used but naive pattern overlay methods, faster than the $2^nn^{O(1)}$ time of a reduction to 3-dimensional matching~\cite{Bjo-STACS-10}, and fast enough to include in a Sudoku solver.
Secondarily, we explain the exponential running time of our solution and of the other previous solutions for this problem by showing that single-digit deduction is NP-hard. 

\section{Nishio}

\def\Nishio{\mathop{\mathrm{Nishio}}}

Define a \emph{valid placement} of an $n\times n$ Sudoku puzzle to be a set of $n$ cells of the puzzle that includes exactly one cell in each row, column, and block of the puzzle, and let $\mathcal{V}$ denote the set of valid placements. Suppose that the initial hints of a puzzle and prior deductions have limited the placements of a digit $d$ to a set $S$ of possible cells (including the ones already known to contain $d$). Then, we may further restrict $S$ to its subset
$$\Nishio(S)=\bigcup\bigl\{V\mid (V\in\mathcal{V})\wedge(V\subset S)\bigr\}.$$
That is, we find the valid placements that are subsets of $S$,  let $\Nishio(S)$ be the set of cells participating in at least one valid placement, and eliminate potential placements for $d$ that do not belong to $\Nishio(S)$. This calculation defines the Nishio deduction rule.

\begin{figure}[t]
\centering\includegraphics[height=1.25in]{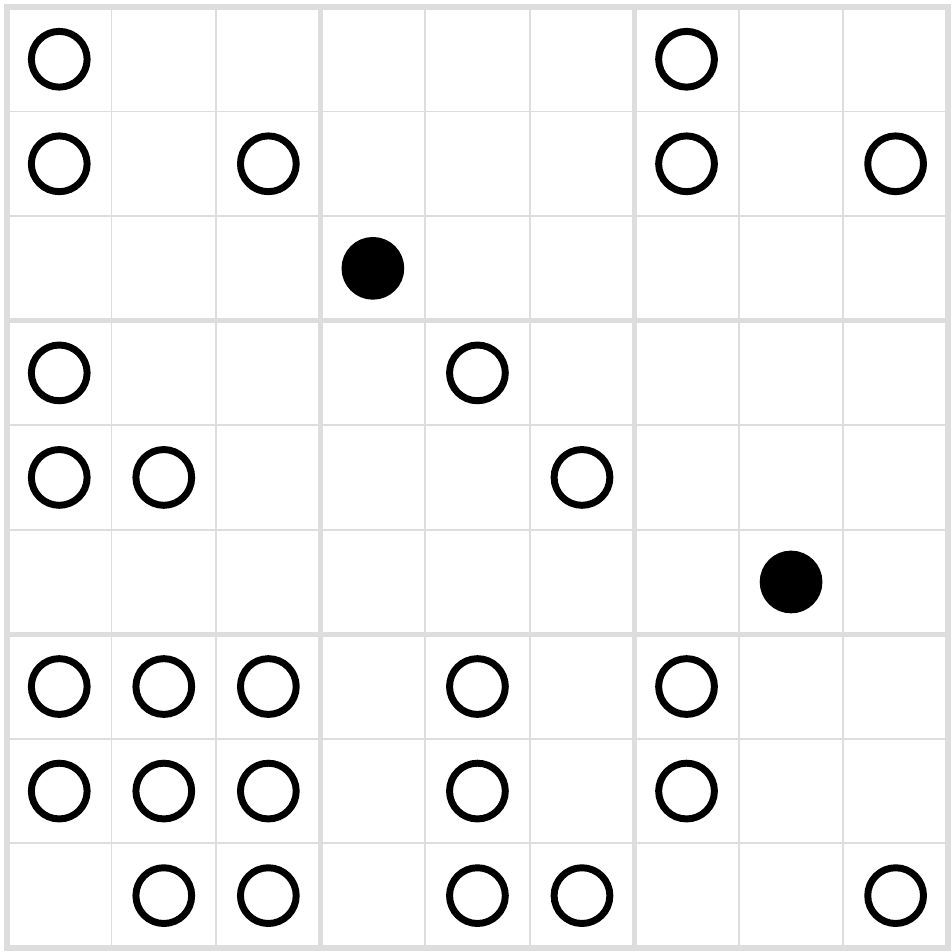}\qquad\includegraphics[height=1.25in]{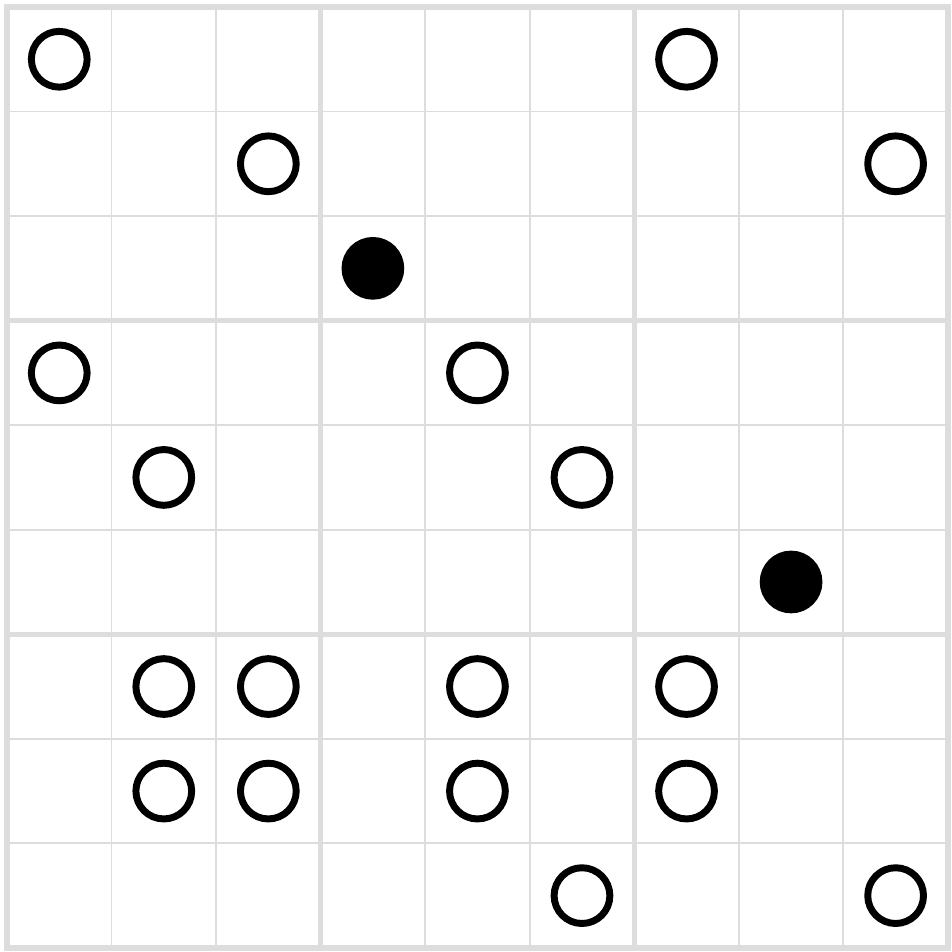}
\caption{The available positions for the digit 7 in the puzzle of Figure~\ref{fig:example} before applying Nishio (left) and after applying Nishio (right). Positions that might possibly hold a 7 are shown as open circles, and positions that must be a 7 are shown as solid disks.}
\label{fig:ex7}
\end{figure}

For instance, consider the Sudoku puzzle in Figure~\ref{fig:example}. After several straightforward deductions, two cells that were initially open to the digit 7 become occupied by other digits, after which the remaining set $S$ of cells in which 7's may be placed is shown in Figure~\ref{fig:ex7}(left). $\Nishio(S)$ is the subset of $S$ shown in Figure~\ref{fig:ex7}(right). 
The cell in the middle of the bottom row (row 9 column 5, or R9C5 for short) is not part of any valid placement that is a subset of $S$, and is therefore not in $\Nishio(S)$. This may be seen by finding the valid placements in~$S$ (there are four of them) and checking that none of them contains this cell. An alternative method is to follow a chain of implications that would follow from placing a 7 in R9C5, using the fact that each $3\times 3$ blocks of the puzzle mus contain the digit~7. If a 7 is placed in R9C5, it would force the placement of a 7 in, successively, R5C6, R4C1, R2C3, R1C7, and R9C9. But then, there would be two 7's in row 9, in the cells R9C5 and R9C9. This contradiction shows that the placement of a 7 into R9C5 is impossible. Only one other digit, 3, can be placed into R9C5, and after this placement another straightforward sequence of deductions completes the puzzle.

Graph matching can be used in a polynomial time approximation to $\Nishio(S)$~\cite{cs.DS/0507053}. Define a bipartite graph $G$ for which the $n$ vertices on one side of the bipartition represent the rows of the puzzle, the $n$ vertices on the other side represent the columns, and in which an edge exists between two vertices $r$ and $c$  if and only if the cell in row $r$ and column $c$ of the puzzle belongs to $S$. The perfect matchings of $G$ correspond to subsets of $S$ that cover each row and column of the puzzle exactly once. If $T$ is the set of edges that participate in at least one perfect matching, then $\Nishio(S)\subseteq T\subseteq S$, and $T$ may be safely used as an approximation to $\Nishio(S)$. However,
by ignoring the requirement that a placement must have exactly one cell in each block, this method allows some invalid placements and therefore may not make as many deductions as the full Nishio rule. In particular, this approximate Nishio method is incapable of solving Figure~\ref{fig:example}.

The calculation of $\Nishio(S)$ may be expressed as an exact satisfiability problem with $3n$ clauses, one for each row, column, and block of the puzzle. Known methods for exact satisfiability of instances with $c$ clauses take time $2^cc^{O(1)}$~\cite{BjoHus-Algo-08,Mad-IPL-06}, and can therefore be used to compute $\Nishio(S)$ in time $8^n n^{O(1)}$. Alternatively it can be expressed as a 3-dimensional matching problem, and solved in randomized expected time $2^n n^{O(1)}$~\cite{Bjo-STACS-10}.

In practice, Nishio deductions are often implemented using an algorithm called the \emph{pattern overlay method}~\cite{Bro-NAW-06,POM}. In this method, one precomputes a list of all the valid placements in $\mathcal{V}$; in $9\times 9$ Sudoku, there are $46656$ of them. Then,
whenever a Nishio deduction is called for, one loops through this list, testing whether each placement is a subset of $S$ and computing the union of the placements that are subsets of $S$. In $n\times n$ Sudoku there are $n!/2^{O(n)}$ valid placements, each of which may be tested in time $O(n^2)$, so the method takes time $O(n!/2^{O(n)})$. Although usable for $9\times 9$ Sudoku, this method is too slow for $16\times 16$ Sudoku, which has more than $10^{11}$ valid placements.

Our goal in this paper is to show how to perform the full Nishio deduction rule while at the same time reducing its time bound to something that, while still exponential, is significantly faster than the previously known methods described above.

\section{Algorithm}
\label{sec:alg}

We now describe our algorithm for calculating $\Nishio(S)$. Our main idea is to construct a directed acyclic graph $G_S$ whose source-to-sink paths correspond one-to-one with the valid placements of the Sudoku puzzle, with the~$i$th edge of a path corresponding to the cell in column~$i$ of a placement. We use depth first search to determine the set of graph edges that participate in source-to-sink paths; the set $\Nishio(S)$ that we wish to calculate is the set of puzzle cells corresponding to these participating graph edges.

The vertices of our graph will correspond to certain sets of rows of a Sudoku puzzle, and it is convenient to represent these sets using $0$--$1$ matrices. For instance, a set $R$ of rows in a standard $9\times 9$ Sudoku puzzle may be represented by a $3\times 3$ matrix
$$\left[ \begin{array}{ccc}
r_1 & r_2 & r_3 \\
r_4 & r_5 & r_6 \\
r_7 & r_8 & r_9
\end{array} \right]$$
where each coefficient is either $0$ or $1$: $r_i=1$ when row~$i$ belongs to $R$, and $r_i=0$ when row~$i$ does not belong to~$R$. Coefficients in the same row of the matrix as each other correspond to rows of the Sudoku puzzle that pass through the same blocks as each other; coefficients that are in different rows of the matrix correspond to puzzle rows that pass through disjoint blocks.
In the same way, if a Sudoku puzzle has blocks of size $a\times b$, we may represent sets of its rows by $0$--$1$ matrices of dimensions $b\times a$ in such a way that two coefficients belong to the same row of a matrix if and only if the corresponding two puzzle rows pass through the same blocks. However, not every set of puzzle rows, and not every $0$-$1$ matrix, will be suitable for use in our graph: we restrict our graph's vertices to the \emph{almost-balanced matrices}, defined below.

\begin{definition}
An \emph{almost-balanced matrix} is a $0$-$1$ matrix in which every two rows have numbers of nonzeros that differ by at most one.
The \emph{count} of a $0$-$1$ matrix is its total number of nonzeros.
\end{definition}

\begin{figure}[t]
\centering\includegraphics[width=\textwidth]{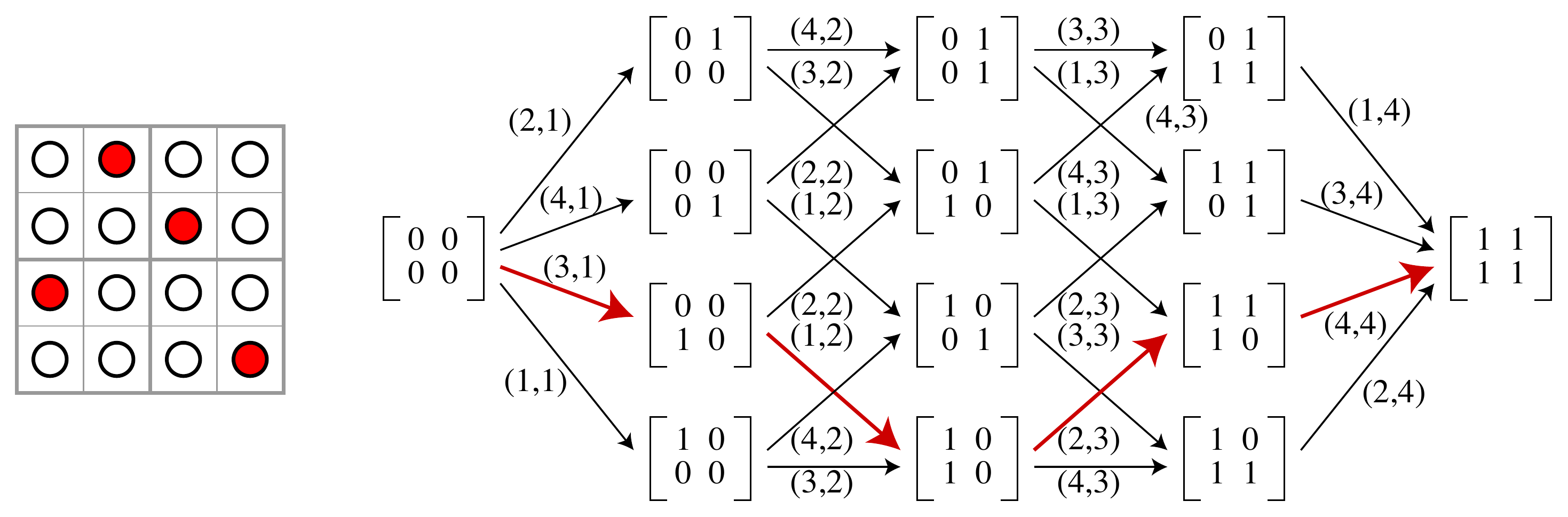}
\caption{Left: A valid placement of four cells among the sixteen cells of a $4\times 4$ Sudoku puzzle grid. Right: the graph $G_S$ of almost-balanced matrices constructed from these sixteen cells, with its edges labeled by puzzle positions, and the correspondence between placements and paths.}
\label{fig:reachability}
\end{figure}

Next we describe the edges of  graph $G_S$. Each edge will correspond to a pair of almost-balanced matrices that differ from each other in a single coefficient. More specifically, suppose that $A$ and $B$ are two almost-balanced matrices, differing from each other in their $r$th coefficient, and let this coefficient be $0$ in $A$ and $1$ in $B$. Let $c$ be the count of matrix $B$. Then, if the cell in row $r$ and column $c$ of the puzzle belongs to $S$, we add an arc in $G_S$ from $A$ to $B$, labeled by the position $(r,c)$. However, if puzzle cell $(r,c)$ does not belong to $S$, we leave vertices $A$ and $B$ disconnected from each other in $G_S$.

\begin{lemma}
\label{lem:path-placement-correspondence}
The valid placements in $S$ are in one-to-one correspondence with the paths in $G_S$ from the all-zero matrix to the all-one matrix.
\end{lemma}

\begin{proof}
Let $P\subset S$ be a valid placement, and let $P_i$ denote the set of rows covered by the cells in the first $i$ columns of $P$. Then we may find a path in $G_S$ in which the $i$th vertex is the matrix corresponding to $P_i$. Each such matrix is balanced, because the cells in the first $i$ columns of $P$ must cover either $\lfloor i/b\rfloor$ or $\lceil i/b\rceil$ of the blocks that intersect each row of the puzzle, and correspondingly each row of the matrix must have either $\lfloor i/b\rfloor$ or $\lceil i/b\rceil$ nonzeros. For each $i$, the two matrices corresponding to $P_i$ and $P_{i+1}$ are connected by an arc, corresponding to the puzzle cell in column $i+1$ of $P$.

In the other direction, any path through $G_S$ defines a set $P$ of cells in $S$, the cells appearing as labels on the path edges. Each edge of the path goes from a matrix of some count $i$ to a matrix of count $i+1$, so the path has to pass through all possible counts and therefore $P$ includes a cell in every column. Each path starts at the all-zero matrix and changes the coefficients of the matrix one by one to all ones, so the path has to pass through edges that change each individual coefficient and therefore $P$ includes a cell in every row. And, because of the balance condition, the part of the path that passes from count $bi$ to count $b(i+1)$ must add exactly one nonzero to each row of the almost-balanced matrices it passes through, from which it follows that $P$ includes one cell in each block in the $i$th column of blocks. Thus, $P$ is a valid placement.
\end{proof}

The construction of the graph $G_S$, and the correspondence between placements and paths, is illustrated for $4\times 4$ Sudoku puzzles with $2\times 2$ blocks in Figure~\ref{fig:reachability}.

To simplify our algorithm for computing $\Nishio(S)$, we precompute the graph $G_U$, where $U$ is the set of all puzzle cells; $G_S$ is the subgraph of edges of $G_U$ whose label belongs to~$S$. We let $\mathbf{0}$ and $\mathbf{1}$ denote the all-zero and all-one matrices respectively (the source and the sink vertices of $G_S$).
Our algorithm performs the following steps:
\begin{enumerate}
\item Initialize a dynamic set $T=\emptyset$ of puzzle cells.
\item Initialize a dynamic set $V=\{\mathbf{1}\}$ of vertices of $G_S$ that no longer need to be searched.
\item Initialize a dynamic set $Z=\{\mathbf{1}\}$ of vertices known to have paths to $\mathbf{1}$.
\item Perform a recursive depth-first search starting from $\mathbf{0}$. For each vertex $v$ reached by the search, add $v$ to $V$, and then perform the following steps for each edge $vw$ in $G_U$ such that the label of $vw$ belongs to $S$:
\begin{enumerate}
\item If $w$ is not in $V$, then search $w$ recursively.
\item If $w$ is in $Z$, then add $v$ to $Z$, and add the label of $vw$ to $T$.
\end{enumerate}
\end{enumerate}
It follows from Lemma~\ref{lem:path-placement-correspondence} that, after performing these steps, the set $T$ constructed by this algorithm is exactly $\Nishio(S)$.

\section{Analysis}

Our algorithm's runtime is dominated by the number of  $a\times b$ almost-balanced matrices, which we abbreviate as  $A_{a,b}$. Although we do not know of a closed-form formula for $A_{a,b}$, it is not difficult to calculate it as a sum over binomial coefficients:

$$A_{a,b}=\sum_{i=0}^{b-1}\left(\binom{b}{i}+\binom{b}{i+1}\right)^a
-\sum_{i=1}^{b-1}\binom{b}{i}^a.$$

Each term in the left sum counts the number of ways of forming an almost-balanced matrix whose $a$ rows each have either $i$ or $i+1$ nonzeros. The almost-balanced matrices in which all rows have the same number $j$ of nonzeros (except for the all-zero and all-one matrices) are counted twice in the left sum, once for $i=j-1$ and a second time for $i=j$. This double counting is corrected  by subtracting the right sum.

\begin{theorem}
The algorithm described in Section~\ref{sec:alg} correctly computes $\Nishio(S)$, for $ab\times ab$ Sudoku puzzles, in time $O(ab\, A_{b,a})$.
\end{theorem}

\begin{proof}
Correctness follows from Lemma~\ref{lem:path-placement-correspondence}. The graph $G_S$ has $A_{b,a}$ vertices; these vertices may be constructed (as an explicit set of matrices) in the stated time bound, by iterating over the possible numbers of nonzeros per row and forming all combinations of rows with the given numbers of nonzeros, following the formula for $A_{a,b}$ above. Each vertex has at most $ab$ outgoing edges, and for each possible outgoing edge we can test whether the corresponding cell belongs to $S$ and if so construct the edge in constant time. Thus, the total number of edges in $G_S$ is at most $ab\, A_{b,a}$, and $G_S$ can be constructed in time $O(ab\, A_{b,a})$. The remaining steps of the algorithm are all also proportional to the size of $G_S$, and therefore also at most $O(ab\, A_{b,a})$.
\end{proof}

For $n\times n$ Sudoku puzzles with square blocks, Stirling's formula may be used to show that each binomial coefficient in the formula for $A_{\sqrt n,\sqrt n}$ is $O(2^{\sqrt n}/n^{1/4})$. Each term in the summation is the sum of two coefficients, taken to the power $\sqrt n$, and is therefore asymptotically bounded by $2^{n-\Omega(\sqrt n\log n)}=o(2^n)$. Multiplication by a polynomial does not change this asymptotic bound, which therefore also applies to $A_{\sqrt n,\sqrt n}$ itself and to the running time of the algorithm.

To compare this to the pattern overlay method, we must count the number of valid placements in an $n\times n$ Sudoku puzzle, again with square blocks. A valid placement may be specified by choosing, for each contiguous group of $\sqrt n$ columns of the puzzle, a one-to-one correspondence from columns to blocks such that corresponding pairs of columns and blocks are covered by the same cell in the valid placement, and by simultaneously choosing in the same way a one-to-one correspondence from rows to blocks. With these two sets of correspondences in hand, the unique cell of the placement within each block can be determined from the row and column corresponding to that block. Therefore, the number of valid placements is exactly $(\sqrt n)!^{2\sqrt n}$~\cite{Dah-LAA-09}. For instance, for $9\times 9$ Sudoku this number is $3!^6=46656$. Asymptotically (via Stirling's formula) this number is $2^{n\log_2 n-O(n)}$. Therefore, the time for the pattern overlay method, which is at least proportional to the number of valid placements, grows more quickly as a function of $n$ than any constant power of the time for our new algorithm.

For $9\times 9$ Sudoku our graph $G_S$ has $290$ vertices and at most $936$ edges, in contrast to the set of $46656$ valid placements, already enough of a difference to predict that our algorithm will be significantly faster than the pattern overlay method. For $16\times 16$ Sudoku the difference is much greater: our graph $G_S$ has $19442$ vertices in this case, whereas the number of valid placements is $4!^8=110075314176$.

\section{Implementation}

\begin{figure}[t]
\centering\includegraphics[height=1.75in]{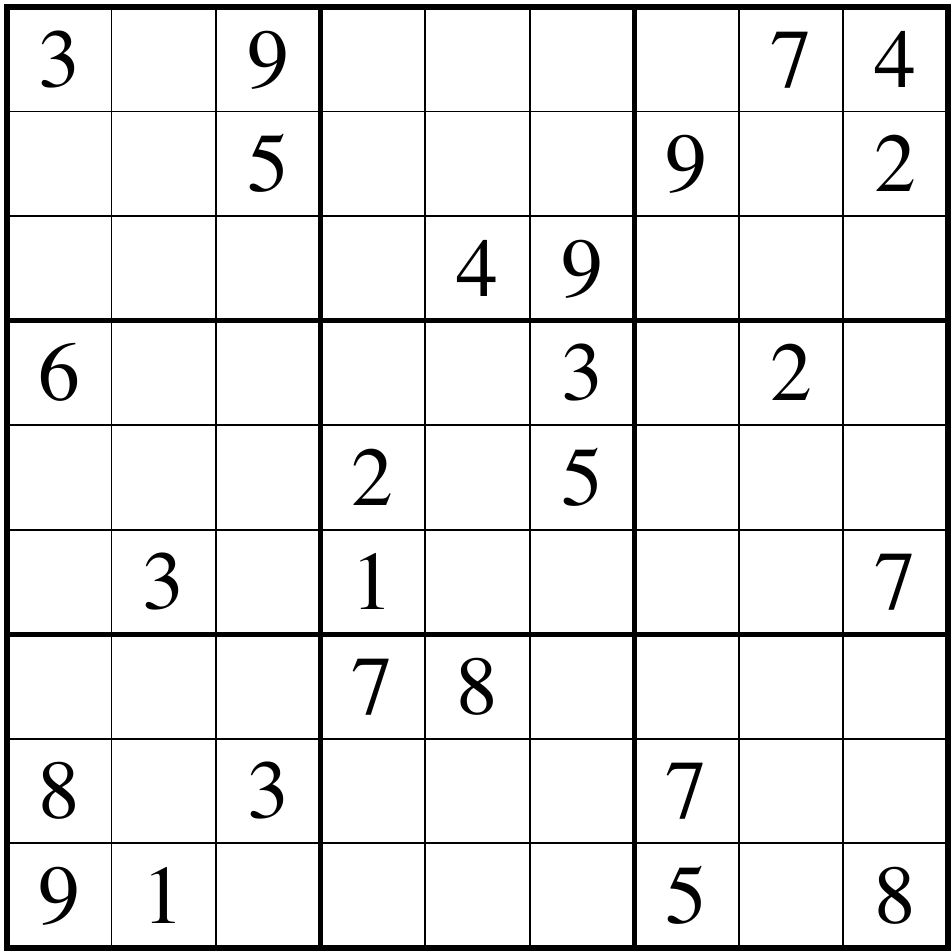}
\caption{A very difficult Sudoku puzzle, formerly impossible for our software without backtracking, now solvable using the new Nishio algorithm.}
\label{fig:possible}
\end{figure}

We have implemented the new algorithm described here, and included it as part of a Python program for solving and generating Sudoku puzzles reported on in our previous paper~\cite{cs.DS/0507053}.\footnote{Available online as the Sudoku package in \url{http://www.ics.uci.edu/~eppstein/PADS/}} We did not carefully time our implementation, because of the difficulty of getting meaningful timings in a slow interpreted language such as Python, but the addition of this deduction rule did not cause any perceptible slowdown to our program.

Our program includes deductive rules that, unlike the Nishio rule described here, combine information about the placements of multiple digits in multiple cells, including the non-repetitive path rules from our previous paper, and a heuristic reduction to 2-SAT. The 2-SAT instances that our software constructs from a Sudoku puzzle also include implications about the placements of single digits that are intended to cover many common Nishio examples. But despite these additional rules, the new Nishio solver has increased the power of our solver. In particular,
among a collection of $36$ Sudoku puzzles that our software previously found impossible to solve without backtracking, one problem shown in Figure~\ref{fig:possible} has now become solvable thanks to the new Nishio solver.

Perhaps more importantly, by using this Nishio solver, the software can now provide simpler human-readable explanations of how to solve difficult Sudoku puzzles. For instance, the puzzle shown in Figure~\ref{fig:example} was previously solvable by our software, but only by means of both the non-repetitive path rules and the reduction to 2-SAT. In the new version of our solver, the only complex rule needed to solve this puzzle is Nishio.

\section{Hardness}

The time bound for our Nishio algorithm is exponential; to justify our failure to provide a polynomial algorithm, we show that $\Nishio(S)$ is NP-hard to compute.
Computing $\Nishio(S)$ is a functional problem, and to determine its computational complexity it is more convenient to reduce it to a decision problem, which we call the \emph{Nishio decision problem}. In this problem, one is given as input two integers $a$ and $b$, and a set $S$ of cells in an $ab\times ab$ grid for Sudoku puzzles with block size $a\times b$. The output is yes if this set of cells contains a valid solution; that is, if $\Nishio(S)\ne\emptyset$; it is no if $\Nishio(S)=\emptyset$.

\begin{theorem}
The Nishio decision problem is NP-complete.
\end{theorem}

\begin{figure}[t]
\centering\includegraphics[width=\textwidth]{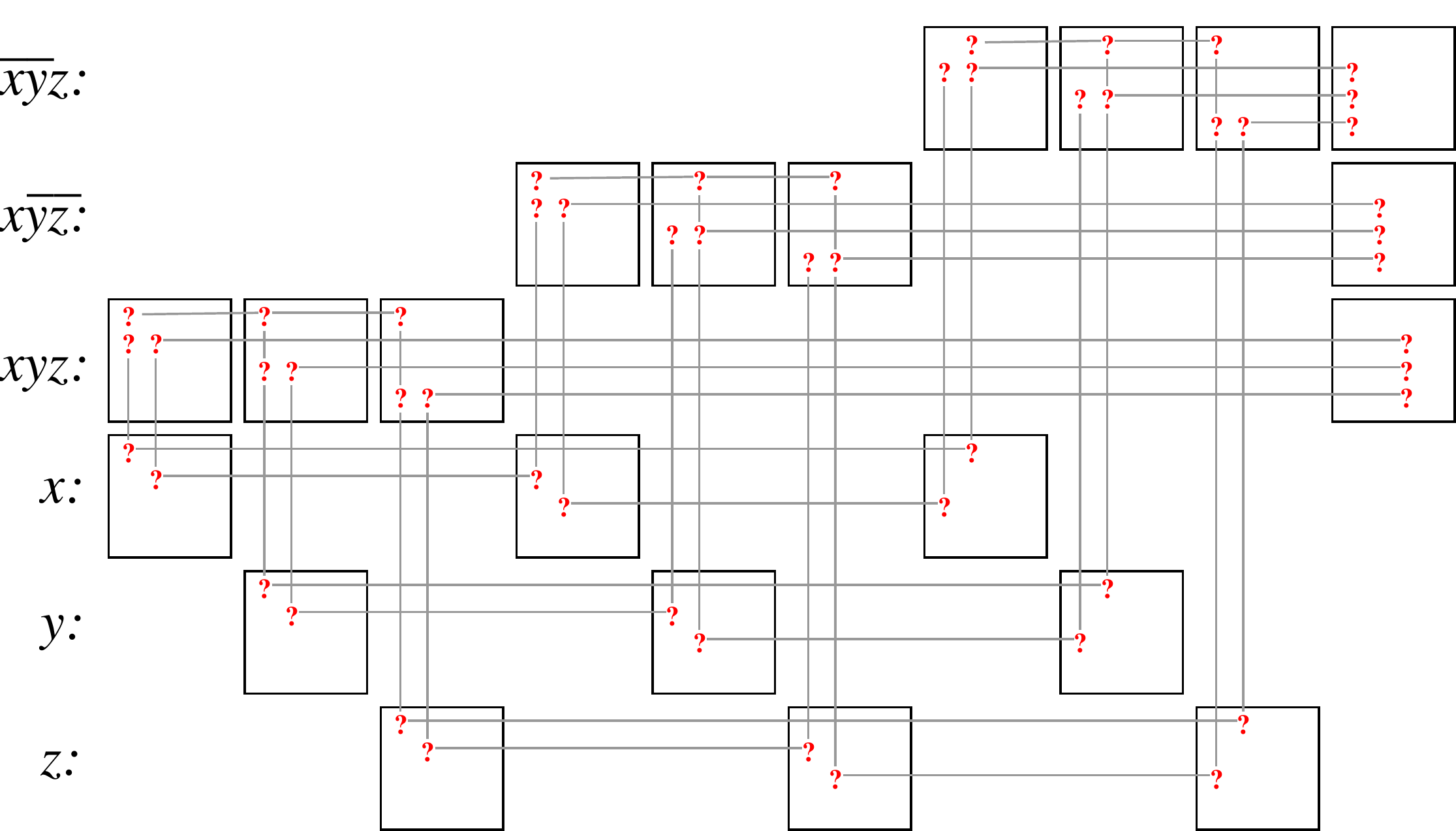}
\caption{Reduction from 3-SAT to the Nishio decision problem}
\label{fig:3sat}
\end{figure}

\begin{proof}
Clearly, the problem is in NP: to show that the output is yes, we need only exhibit a single valid placement that is a subset of $S$.

To show NP-hardness, we reduce 3-SAT to the Nishio decision problem, by transforming any 3-SAT instance into an instance $S$ of the Nishio problem. The transformed problem uses an $n^2\times n^2$ grid with block size $n\times n$ where $n$ is chosen to be at least as large as $\max(v+c,3c+1)$ and where $v$ and $c$ are the numbers of variables and clauses in the 3-SAT instance. We create a gadget for each 3-SAT variable consisting of certain cells in a horizontally aligned set  of Sudoku blocks. Each of these blocks contains two cells of $S$, in two distinct rows and columns of the grid, and each of these cells is in the same row as one cell in a different block of the gadget, so that the conflict graph of all these cells forms a single  cycle. There are only two consistent ways of placing a digit within these blocks and rows: one using the even positions in the cycle (which we associate with assignments of the value true to the variable)  and the other using the odd positions (which we associate with assignments of the value false to the variable).

We also make a gadget for each 3-SAT clause consisting of cells in a horizontally aligned set of four blocks, three for the three terms of the clause and one more terminal block. We place into $S$ three cells in each of these four blocks, lying on four rows that pass through these four blocks, one ``truth-setting row'' and three ``conflict rows''. In the block corresponding to a term $t$, $S$ contains one cell in the truth-setting row, and two cells in the conflict row for that block; all three of these cells are placed within the same two columns as the cells in the corresponding block of the variable gadget for $t$. These cells are placed in such a way that, no matter which truth assignment is used for $t$, one of the two positions in the conflict row for the clause remains available. However, the choice of the column on which to place the cell in the truth-setting row of the block for term $t$ is made in such a way that this cell may only be used when $t$ is given a truth assignment compatible with the clause. In the terminal block, all three cells are chosen within a single column, one in each of the gadget's three conflict rows.

Finally, in order to ensure that the set of cells in these gadgets can be part of a valid placement for the whole puzzle, we add to $S$ any cell that is outside the union of the blocks, rows, and columns used in these gadgets.

In any valid placement that is a subset of $S$, the cells of the placement within each variable gadget must correspond to a valid truth assignment to that variable, for otherwise not all of the gadget's blocks could be covered by the placement. In each clause gadget, the terminal block must have one of its cells chosen, on one of the conflict rows, and the corresponding term block must have a cell chosen on the truth-setting row, compatibly with the truth assignment to that term. Therefore, any valid placement must necessarily correspond to a truth assignment for the given 3-SAT instance.

In the other direction, suppose that our given 3-SAT instance has a valid truth assignment; that is, for each clause we may choose a term in that clause that is assigned to be true. We may use this truth assignment, and the choice of the true term in each clause, to guide the choice of a non-conflicting set of cells in each row, column, and block of the gadgets of the Nishio instance. Note that this choice necessarily includes a cell within each of the rows, columns, and blocks used by the gadgets. The remaining rows, columns, and blocks of the puzzle grid may be covered by cells chosen by a greedy algorithm as there is no possibility of additional conflicts.

Thus, the two problems, the initial 3-SAT instance and the Nishio decision problem derived from it, are equivalent: one problem has a positive solution if and only if the other does. This equivalence completes the NP-completeness proof.
\end{proof}

The gadgets used in this reduction are depicted in schematic form in Figure~\ref{fig:3sat}. It follows immediately that computing $\Nishio(S)$ for arbitrary sets $S$ is NP-hard.
Conceptually, however, there is a minor difficulty with applying this result to Sudoku puzzles,
because whenever we apply the Nishio rule to a Sudoku puzzle we will be guaranteed that $\Nishio(S)$ is nonempty. Thus, the class of sets for which it is difficult to tell whether $\Nishio(S)$ is empty and the class of sets arising in puzzle solving may be disjoint. 
To resolve this issue, we use an additional trick.

\begin{theorem}
It is NP-hard to compute $\Nishio(S)$, even restricted to the class of instances for which $\Nishio(S)$ is nonempty.
\end{theorem}

\begin{proof}
Starting with a hard 3-SAT instance $X$, we modify $X$ to produce a new 3-SAT instance $X'$ with an additional variable $v_0$, such that there always exists a solution to $X'$ with $v_0$ true but such that there exist solutions with $v_0$ false if and only if $X$ is satisfiable.
To do so, we turn $X$ into a 4-SAT instance by adding $v_0$ to each of its clauses,
and then split every four-term clause $(a\vee b\vee c\vee d)$ into two three-term clauses
$(a\vee b\vee x)\wedge (\lnot x\vee c\vee d)$ where $x$ is a new variable introduced for each clause that needs to be split.

We then use the same reduction as the one above, on the modified instance $X'$, and let $S$ be the set of cells constructed by the reduction. Then $\Nishio(S)$ will always be nonempty (it will contain valid placements corresponding to truth assignments with $v_0$ true). However, it will contain the cells from the false assignment to $v_0$ if and only if the initial 3-SAT instance $X$ is satisfiable. Because it is NP-hard to determine whether $X$ is satisfiable, it is also NP-hard to determine whether these cells belong to $\Nishio(S)$.
\end{proof}

\section{Conclusions}

We have developed a new algorithm for solving single-digit deduction  problems in Sudoku. Despite the NP-hardness of the problem, our algorithm solves it in an amount of time that is only mildly exponential, significantly faster than the pattern overlay method previously used for the same task. Our algorithm does not explain how humans might solve Nishio problems (itself an interesting problem for future research) but can be used as one of the deduction rules in a human-like non-backtracking Sudoku solver.

The requirements that a valid placement cover each row and column of the puzzle exactly once can be formulated in terms of bipartite matching. but the additional constraints coming from the blocks of the puzzle form an obstacle to the polynomial-time solution of this sort of matching problem. Matching has many other applications, and it is possible that the techniques we develop here may be of interest in some of them. Our method also resembles dynamic programs for the traveling salesman problem and related problems~\cite{Bel-JACM-62,HelKar-JSIAM-62} and might potentially help speed up some of these problems.

A curiosity arising from this work is that the graph of $2\times 2$ almost-balanced matrices in Figure~\ref{fig:reachability} is the skeleton of the rhombic dodecahedron. Is this a coincidence, or do the larger graphs constructed in the same way have a similar polyhedral description?

The graph reachability algorithm at the heart of our  method takes time linear in the number of edges of  graph~$G_S$. In other application domains where graph reachability algorithms also apply, we have obtained additional speedups using bit-parallel programming techniques~\cite{Epp-MSRI-02,Epp-JGAA-11}. Can similar techniques also speed up Nishio deduction?

Another question in the nexus of Sudoku puzzle solving and exponential time exact algorithms remains unsolved: what is the worst case time for solving complete $n\times n$ Sudoku puzzles? Is it possible to achieve $2^{o(n^2)}$ time for this problem?

\raggedright
\bibliographystyle{abuser}
\bibliography{sudoku}
\end{document}